\documentclass{iopart}
\usepackage{graphics}
\graphicspath{ {figures/} }

\usepackage{amsopn}
\usepackage{iopams}

\usepackage{hyperref}

\hypersetup{
 bookmarks=true,		
 unicode=false,			
 pdftoolbar=true,		
 pdfmenubar=true,		
 pdffitwindow=false,		
 pdfstartview={FitH},		
 pdftitle={The classification of multipartite quantum correlation},    
 pdfauthor={Szil{\'a}rd Szalay},	
 pdfsubject={research article on multipartite quantum correlation and entanglement},	
 pdfcreator={pdflatex},		
 pdfproducer={vim},		
 pdfkeywords={quantum} {entanglement} {correlation} {multipartite} {mixed states}, 
 pdfnewwindow=true,		
 colorlinks=true,		
 linktoc=page,			
 linkcolor=blue,		
 citecolor=blue,		
 filecolor=blue,		
 urlcolor=blue			
}

\newcommand{\vs}[1]{\boldsymbol{#1}}
\newcommand{\vvs}[1]{\underline{\boldsymbol{#1}}}

\providecommand{\abs}[1]{\vert#1\vert}

\newcommand{\cmpl}[1]{\overline{#1}}

\DeclareMathOperator{\Lin}{Lin}
\DeclareMathOperator{\Conv}{Conv}

\DeclareMathOperator{\downset}{\downarrow}
\DeclareMathOperator{\upset}{\uparrow}

\newcommand{\inlineiff}{\,\Leftrightarrow\,}

\newcommand{\inlinethen}{\,\Rightarrow\,}

\newcommand{\displayiff}{\;\;\Longleftrightarrow\;\;}

\newcommand{\displaythen}{\;\;\Longrightarrow\;\;}

\newcommand{\set}[1]{\{ #1 \}}

\newcommand{\sset}[2]{\{\, #1 \,\vert\, #2 \,\}}
\newcommand{\bigsset}[2]{\bigl\{\, #1 \,\big\vert\, #2 \,\bigr\}}
\newcommand{\Bigsset}[2]{\Bigl\{\, #1 \,\Big\vert\, #2 \,\Bigr\}}

\newcommand{\struct}[1]{( #1 )}

\newcommand{\text}[1]{\mathrm{#1}} 

\newtheorem{prop}{Proposition}
\newtheorem{lem}[prop]{Lemma}
\newenvironment{proof}{\paragraph{\textit{Proof:}}}{\hfill$\square$}
\newenvironment{proofl}[1]{\paragraph{\textit{Proof #1:}}}{\hfill$\square$}
 
\begin{document}
\title{The classification of multipartite quantum correlation}
\author{{Sz}il{\'a}rd {Sz}alay}
\address{Strongly Correlated Systems ``Lend{\"u}let'' Research Group,\\
Wigner Research Centre for Physics of the Hungarian Academy of Sciences,\\
29-33, Konkoly-Thege Mikl{\'o}s {\'u}t, Budapest, H-1121, Hungary}
\ead{szalay.szilard@wigner.mta.hu}
\date{\today}

\begin{abstract}
In multipartite entanglement theory,
the partial separability properties have an elegant, yet complicated structure,
which becomes simpler in the case when multipartite correlations are considered.
In this work, we elaborate this,
by giving 
necessary and sufficient conditions for the existence and uniqueness of the class of a given class-label,
by the use of which we work out the structure of the classification for some important particular cases,
namely, 
for the finest classification,
for the classification based on $k$-partitionability and $k$-producibility,
and for the classification based on the atoms of the correlation properties.
\end{abstract}



\section{Introduction}
\label{sec:Intro}

In quantum systems, nonclassical forms of correlations arise, 
which, although being simple consequences of the Hilbert space structure of quantum mechanics,
represent a longstanding challenge for the classically thinking mind.
\emph{Pure states} of classical systems are always uncorrelated;
correlations in pure states are of quantum origin,
this is what we call entanglement~\cite{Schrodinger-1935a,Horodecki-2009}.
The correlation in \emph{mixed states} of classical systems can be induced by classical communication;
correlations in mixed states which are not of this kind are of quantum origin, 
this is what we call entanglement~\cite{Werner-1989,Horodecki-2009}.

\emph{Bipartite} systems can either be uncorrelated or correlated,
and either be separable or entangled,
while for \emph{multipartite} systems,
the \emph{partial separability} properties have a complicated, yet elegant structure
\cite{Dur-1999,Dur-2000,Acin-2001,Nagata-2002,Seevinck-2008,
Szalay-2012,Szalay-2015b,Szalay-2017}.
Considering the \emph{partial correlation} properties of multipartite systems~\cite{Szalay-2017},
the structure of the classification~\cite{Szalay-2015b} becomes simpler.
In the present work, we elaborate this,
by giving
necessary and sufficient conditions 
for the existence and uniqueness of the class of a given class-label,
by the use of which we elaborate the structure of the classification in some important particular cases.

Our work is motivated by that
quantum correlation and entanglement are of central importance 
in many fields of research in quantum physics nowadays,
first of all
in quantum information theory~\cite{Nielsen-2000,Petz-2008,Wilde-2013}
and in strongly correlated manybody systems~\cite{Amico-2008,Legeza-2004,Szalay-2015a}.
Especially in the latter case,
correlation might be more important than entanglement,
since in physical properties of manybody systems,
the entire correlation is what matters, not only its entanglement part.
(The two coincide only for pure states, so almost never inside subsystems.)
Fortunately, this also meets the claim of practice, since
the measures of multipartite correlations are feasible to evaluate~\cite{Szalay-2015b,Szalay-2017},
while this is not the case for the measures of multipartite entanglement~\cite{Plenio-2007,Eltschka-2014,Szalay-2015b}.

The organization of the paper is as follows.
In section~\ref{sec:recall} we recall the structure of multipartite correlation and entanglement,
ending in the definitions of the partial correlation classes,
which are labelled by a natural labelling scheme.
This formalism allows us to describe all the possible partial correlation based classifications.
Because of the structure of the partial correlation properties,
this labelling scheme, while being conjectured to be faithful for partial entanglement,
is not faithful for partial correlations:
on the one hand, there are labels which define empty classes,
on the other hand, different labels may lead to the same partial correlation class.
In section~\ref{sec:struct} we elaborate this, by giving general necessary
and sufficient conditions for the existence and uniqueness of the class of a given label.
In section~\ref{sec:corrclassxmpl} we apply our results to some important classifications.
The natural way of the description of the classification is 
the use of the tools of elementary set and order theory (in the finite setting)~\cite{Davey-2002,Roman-2008}.
For the convenience of the reader, 
we recall the elements needed in~\ref{appsec:posets.defs}.
The proofs of some auxiliary results are given in appendices.

\section{Multipartite correlation and entanglement}
\label{sec:recall}

In this section
we briefly recall and slightly extend the results
about the structure of multipartite correlations and entanglement~\cite{Szalay-2015b,Szalay-2017}.
When we go beyond our previous works (\cite{Szalay-2015b} and the supplementary material of \cite{Szalay-2017}),
we give the proofs inline, or in appendices.

\subsection{Level 0: subsystems}
\label{sec:recall.L0}

Let $L=\set{1,2,\dots,n}$ be the set of the labels of the \emph{elementary subsystems.}
All the \emph{subsystems} are then labelled by subsets $X\subseteq L$,
the set of which, $P_0=2^L$, naturally possesses
a Boolean lattice structure with respect to the inclusion $\subseteq$. 
For each elementary subsystem $i\in L$,
we have finite dimensional Hilbert spaces $\mathcal{H}_i$ associated with it ($1<\dim\mathcal{H}_i<\infty$);
from these, the Hilbert space associated with
every subsystem $X\in P_0$ is
$\mathcal{H}_X = \bigotimes_{i\in X}\mathcal{H}_i$.
The \emph{state} of the subsystem $X\in P_0$
is given by a density operator (positive semidefinite operator of trace $1$)
acting on $\mathcal{H}_X$;
the set of the states of subsystem $X$ is denoted with $\mathcal{D}_X$.

\subsection{Level~I: partitions}
\label{sec:recall.LI}

For handling the different possible splits of a composite system into subsystems,
we need to use the mathematical notion of \emph{partition} of the system $L$,
which are sets of subsystems
$\xi=\set{X_1,X_2,\dots,X_{\abs{\xi}}}$, 
where the \emph{parts} $X\in\xi$ are nonempty disjoint subsystems,
for which $\cup\xi:=\bigcup_{X\in\xi}X=L$.
The set of the partitions of $L$ is denoted with $P_\text{I}$
(its size is given by the \emph{Bell numbers}~\cite{oeisA000110}),
it possesses a lattice structure with respect to the \emph{refinement} $\preceq$,
which is the natural partial order over the partitions,
defined as
$\upsilon\preceq\xi$ if for all $Y\in\upsilon$ there is an $X\in\xi$ such that $Y\subseteq X$.
(For illustration, see figure~\ref{fig:Ps3}.)
\begin{figure}\centering
\includegraphics{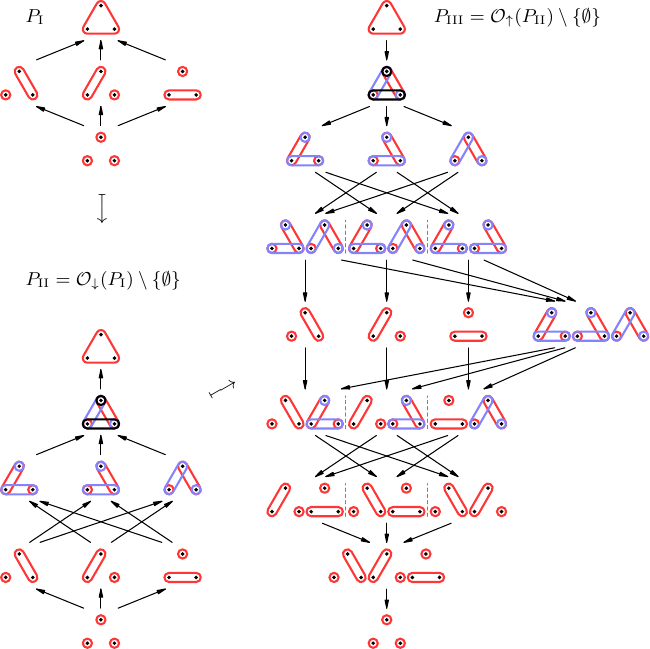}
\caption{The lattices of the three-level structure of multipartite correlation and entanglement for $n=3$.
Only the maximal elements of the down-sets of $P_\text{I}$ are shown (with different colors) in $P_\text{II}$, while
only the minimal elements of the up-sets of $P_\text{II}$ are shown (side by side) in $P_\text{III}$.
The partial orders $\preceq$ are represented by consecutive arrows.}\label{fig:Ps3}
\end{figure}

For a partition $\xi\in P_\text{I}$, the \emph{$\xi$-uncorrelated states} are those which are of the product form
with respect to the partition $\xi$,
\begin{equation}
\label{eq:DuncI}
\mathcal{D}_{\xi\text{-unc}}:=\Bigsset{\varrho_L\in\mathcal{D}_L}
{\forall X\in\xi, \exists \varrho_X\in\mathcal{D}_X:\varrho_L=\bigotimes_{X\in\xi}\varrho_X};
\end{equation}
the others are the \emph{$\xi$-correlated states}.
The \emph{$\xi$-separable states} are those, 
which are convex combinations (statistical mixtures) of $\xi$-uncorrelated ones,
\begin{equation}
\label{eq:DsepI}
\mathcal{D}_{\xi\text{-sep}}:=\Conv\mathcal{D}_{\xi\text{-unc}};
\end{equation}
the others are the \emph{$\xi$-entangled states}.
(The \emph{convex hull} of $A$ is $\Conv A = \sset{\sum_i p_i a_i}{a_i\in A, 0\leq p_i, \sum_i p_i = 1}$.)
These properties show the same lattice structure as the partitions~\cite{Szalay-2015b}, $P_\text{I}$,
that is,
\begin{equation}
\label{eq:DincI}
\upsilon\preceq\xi \displayiff 
\mathcal{D}_{\upsilon\text{-unc}}\subseteq\mathcal{D}_{\xi\text{-unc}},\;
\mathcal{D}_{\upsilon\text{-sep}}\subseteq\mathcal{D}_{\xi\text{-sep}}.
\end{equation}
(For the proof, see~\ref{appsec:qstates.I-II}.)
Note that
$\mathcal{D}_{\xi\text{-unc}}$ is closed under LO (local operations), and
$\mathcal{D}_{\xi\text{-sep}}$ is closed under LOCC (local operations and classical communications \cite{Bennett-1996b})~\cite{Szalay-2015b}.
(Here locality can be considered with respect to $\xi$,
but later this will be restricted to the finest split, $\bot=\sset{\{i\}}{i\in L}$.
The LO closedness, although not being proven in~\cite{Szalay-2015b}, is obvious.)

\subsection{Level~II: multiple partitions}
\label{sec:recall.LII}

The order isomorphism~\eref{eq:DincI} tells us that
if we consider states uncorrelated (or separable) with respect to a partition,
then we automatically consider states uncorrelated (or separable) with respect to every finer partition.
On the other hand, in multipartite entanglement theory,
it is necessary to handle mixtures of states uncorrelated with respect to different partitions~\cite{Acin-2001,Seevinck-2008,Szalay-2015b}.
Because of these, 
for the labelling of the different \emph{partial correlation and entanglement properties,}
we need to use the nonempty \emph{down-sets} of partitions (also called nonempty \emph{ideals} of partitions)~\cite{Szalay-2015b}, which are
sets of partitions
$\vs{\xi}=\set{\xi_1,\xi_2,\dots,\xi_{\abs{\vs{\xi}}}}\subseteq P_\text{I}$,
which are closed downwards with respect to $\preceq$
(that is, if $\xi\in\vs{\xi}$, then every $\upsilon\preceq\xi$ is also $\upsilon\in\vs{\xi}$).
The set of the nonempty partition ideals of $L$ is denoted with $P_\text{II} :=\mathcal{O}_\downarrow(P_\text{I})\setminus\{\emptyset\}$,
it possesses a lattice structure with respect to the standard inclusion
as partial order,
$\vs{\upsilon}\preceq\vs{\xi}$ if and only if $\vs{\upsilon}\subseteq\vs{\xi}$.
(For illustration, see figure~\ref{fig:Ps3}.)
Special cases are the ideals of \emph{$k$-partitionable} and \emph{$k'$-producible partitions},
$\vs{\mu}_k    := \bigsset{\mu\in P_\text{I}}{\abs{\mu}\geq k}$, 
$\vs{\nu}_{k'} := \bigsset{\nu\in P_\text{I}}{\forall N\in\nu: \abs{N}\leq k'}$,
for $1\leq k,k'\leq\abs{L}$,
that is, which contain partitions
where the number of parts is at least $k$,
and where the sizes of the parts are at most $k'$, respectively.
These form chains in the lattice $P_\text{II}$, as
$\vs{\mu}_l    \preceq \vs{\mu}_k \inlineiff l \geq k$, and 
$\vs{\nu}_{l'} \preceq \vs{\nu}_{k'}   \inlineiff l'\leq k'$.

For an ideal $\vs{\xi}\in P_\text{II}$, the \emph{$\vs{\xi}$-uncorrelated states} are those which are $\xi$-uncorrelated 
with respect to a $\xi\in\vs{\xi}$,
\begin{equation}
\label{eq:DuncII}
\mathcal{D}_{\vs{\xi}\text{-unc}}:=\bigcup_{\xi\in\vs{\xi}}\mathcal{D}_{\xi\text{-unc}};
\end{equation}
the others are the \emph{$\vs{\xi}$-correlated states}.
The \emph{$\vs{\xi}$-separable states} are those, 
which are convex combinations of $\vs{\xi}$-uncorrelated ones,
\begin{equation}
\label{eq:DsepII}
\mathcal{D}_{\vs{\xi}\text{-sep}}:=\Conv\mathcal{D}_{\vs{\xi}\text{-unc}};
\end{equation}
the others are the \emph{$\vs{\xi}$-entangled states}.
These properties show the same lattice structure as the partition ideals~\cite{Szalay-2015b}, $P_\text{II}$,
that is,
\begin{equation}
\label{eq:DincII}
\vs{\upsilon}\preceq\vs{\xi} \displayiff 
\mathcal{D}_{\vs{\upsilon}\text{-unc}}\subseteq\mathcal{D}_{\vs{\xi}\text{-unc}},\;
\mathcal{D}_{\vs{\upsilon}\text{-sep}}\subseteq\mathcal{D}_{\vs{\xi}\text{-sep}}.
\end{equation}
(For the proof, see~\ref{appsec:qstates.I-II}.)
Note that
$\mathcal{D}_{\vs{\xi}\text{-unc}}$ is closed under LO, and
$\mathcal{D}_{\vs{\xi}\text{-sep}}$ is closed under LOCC~\cite{Szalay-2015b}.
(Here locality is understood with respect to the finest partition.)
Special cases are
the \emph{$k$-partitionably uncorrelated} and
the \emph{$k'$-producibly uncorrelated} states,
$\mathcal{D}_{k\text{-part\, unc}} := \mathcal{D}_{\vs{\mu}_k\text{-unc}}$ and
$\mathcal{D}_{k'\text{-prod\, unc}}:= \mathcal{D}_{\vs{\nu}_{k'}\text{-unc}}$,
which are of the product form
of at least $k$ density operators,
and of density operators of at most $k'$ elementary subsystems, respectively. 
The \emph{$k$-partitionably separable} (also called \emph{$k$-separable}~\cite{Acin-2001,Guhne-2005,Seevinck-2008,Facchi-2006,Facchi-2010a}) and
the \emph{$k'$-producibly separable} (also called \emph{$k'$-producible}~\cite{Seevinck-2001,Guhne-2005,Toth-2010}) states are
$\mathcal{D}_{k\text{-part\,sep}} := \mathcal{D}_{\vs{\mu}_k\text{-sep}}$ and
$\mathcal{D}_{k'\text{-prod\,sep}}:=\mathcal{D}_{\vs{\nu}_{k'}\text{-sep}}$,
which can be decomposed into $k$-partitionably, and $k'$-producibly uncorrelated states, respectively.
These properties show the same lattice structure (chain) as the corresponding partition ideals,
that is,
$l  \geq k   \inlineiff
  \mathcal{D}_{l\text{-part\,unc}} \subseteq \mathcal{D}_{k\text{-part\,unc}},\;
  \mathcal{D}_{l\text{-part\,sep}} \subseteq \mathcal{D}_{k\text{-part\,sep}}$,
and
$l' \leq k'  \inlineiff  
  \mathcal{D}_{l'\text{-prod\,unc}} \subseteq \mathcal{D}_{k'\text{-prod\,unc}},\;
  \mathcal{D}_{l'\text{-prod\,sep}} \subseteq \mathcal{D}_{k'\text{-prod\,sep}}$,
that is,
if a state is $l$-partitionably uncorrelated (or separable)
then it is also $k$-partitionably uncorrelated (or separable) for all $l\geq k$,
and 
if a state is $l'$-producibly uncorrelated (or separable)
then it is also $k'$-producibly uncorrelated (or separable) for all $l'\leq k'$.

\subsection{Level~III: classes}
\label{sec:recall.LIII}

The partial correlation and entanglement properties form an inclusion hierarchy~\eref{eq:DincII}.
For handling the possible \emph{partial correlation} and \emph{partial entanglement classes}
(which are state-sets of \emph{well-defined} Level~II partial correlation and entanglement \emph{properties},
that is, the possible intersections of the state-sets 
$\mathcal{D}_{\vs{\xi}\text{-unc}}$ and $\mathcal{D}_{\vs{\xi}\text{-sep}}$),
we need to use the nonempty \emph{up-sets} of nonempty down-sets of partitions 
(also called nonempty \emph{filters} of nonempty partition ideals)~\cite{Szalay-2015b}, which are
sets of partition ideals
$\vvs{\xi}=\set{\vs{\xi}_1,\vs{\xi}_2,\dots,\vs{\xi}_{\abs{\vvs{\xi}}}}\subseteq P_\text{II}$ 
which are closed upwards with respect to $\preceq$
(that is, if $\vs{\xi}\in\vvs{\xi}$ then every $\vs{\upsilon}\succeq\vs{\xi}$ is also $\vs{\upsilon}\in\vvs{\xi}$).
The set of the nonempty filters of nonempty partition ideals of $L$ is denoted with 
$P_\text{III} :=\mathcal{O}_\uparrow(P_\text{II})\setminus\{\emptyset\}$,
it possesses a lattice structure with respect to the standard inclusion
as partial order,
$\vvs{\upsilon}\preceq\vvs{\xi}$ if and only if $\vvs{\upsilon}\subseteq\vvs{\xi}$.
(For illustration, see figure~\ref{fig:Ps3}.)
In the generic case,
if the inclusion of sets can be described by a poset $P$,
then $\mathcal{O}_\uparrow(P)$ is \emph{sufficient} for the description of the intersections.
(For the proof, see~\ref{appsec:posets.int}.)
One may make the classification coarser \cite{Szalay-2015b}
by selecting a \emph{sub(po)set} of partial correlation and entanglement properties $P_\text{II*}\subseteq P_\text{II}$,
with respect to which the classification is done,
$P_\text{III*} :=\mathcal{O}_\uparrow(P_\text{II*})\setminus\{\emptyset\}$.
(This is not a lattice if $P_\text{II*}$ has no top element.)

For a filter $\vvs{\xi}\in P_\text{III*}$, the \emph{strictly $\vvs{\xi}$-separable states} are those which are
$\vs{\xi}$-separable for all $\vs{\xi} \in\vvs{\xi}$,
and $\vs{\xi}'$-entangled for all $\vs{\xi}'\in\cmpl{\vvs{\xi}}=P_\text{II*}\setminus\vvs{\xi}$~\cite{Szalay-2015b},
the \emph{class} of these is
\begin{equation}
\label{eq:Csep}
\mathcal{C}_{\vvs{\xi}\text{-sep}} :=
 \bigcap_{\vs{\xi}'\in\cmpl{\vvs{\xi}}} \cmpl{\mathcal{D}_{\vs{\xi}'\text{-sep}}} \cap 
 \bigcap_{\vs{\xi}\in       \vvs{\xi} }       \mathcal{D}_{\vs{\xi}\text{-sep}}.
\end{equation}
(Note that the complement $\cmpl{\vvs{\xi}}$ is always taken with respect to $P_\text{II*}$.)
It is conjectured that 
$\vvs{\xi}$-separability is nontrivial for all $\vvs{\xi}\in P_\text{III*}$
(that is, $\mathcal{C}_{\vvs{\xi}\text{-sep}}$ is nonempty)~\cite{Szalay-2015b}.
Note that 
the Level~III hierarchy $\preceq$ \emph{compares the strength of entanglement} among the classes labelled by $P_\text{III*}$,
in the sense that
if there exists a $\varrho \in \mathcal{C}_{\vvs{\upsilon}\text{-sep}}$ and an LOCC map
mapping it into $\mathcal{C}_{\vvs{\xi}\text{-sep}}$, then $\vvs{\upsilon}\preceq\vvs{\xi}$~\cite{Szalay-2015b}.

If we consider the \emph{class of strictly $\vvs{\xi}$-uncorrelated states},
being the possible intersections of the state sets $\mathcal{D}_{\vs{\xi}\text{-unc}}$,
encoded by the filter $\vvs{\xi}\in P_\text{III*}$ as
\begin{equation}
\label{eq:Cunc}
\mathcal{C}_{\vvs{\xi}\text{-unc}} :=
 \bigcap_{\vs{\xi}'\in\cmpl{\vvs{\xi}}} \cmpl{\mathcal{D}_{\vs{\xi}'\text{-unc}}} \cap 
 \bigcap_{\vs{\xi}\in       \vvs{\xi} }       \mathcal{D}_{\vs{\xi}\text{-unc}}
\end{equation}
(that is, a state is $\vvs{\xi}$-uncorrelated,
if it is $\vs{\xi}$-uncorrelated for all $\vs{\xi} \in\vvs{\xi}$,
and $\vs{\xi}'$-correlated for all $\vs{\xi}'\in\cmpl{\vvs{\xi}}$),
then the structure $P_\text{III*}$ becomes simpler.
In the following section, 
we elaborate this, by giving necessary and sufficient conditions 
for the nonemptiness of the classes
and for the uniqueness of the labels.
Note that if there exists a $\varrho \in \mathcal{C}_{\vvs{\upsilon}\text{-unc}}$ and an LO map
mapping it into $\mathcal{C}_{\vvs{\xi}\text{-unc}}$, then $\vvs{\upsilon}\preceq\vvs{\xi}$.
(This can be proven analogously to the partial separability result with LOCC above,
see Appendix A.12 in~\cite{Szalay-2015b},
it relies only on the LO closedness of the state sets $\mathcal{D}_{\vs{\xi}\text{-unc}}$.)
In this sense, the Level~III hierarchy $\preceq$ \emph{compares the strength of correlation} among the classes labelled by $P_\text{III*}$.

\section{The structure of the classification of correlations}
\label{sec:struct}

In this section,
after establishing some important facts 
about the Level~I-II structure of multipartite correlations 
(section~\ref{sec:struct.corr}),
we give necessary and sufficient conditions
for the existence (section~\ref{sec:struct.corrclassex}) 
and uniqueness (section~\ref{sec:struct.corrclassunique}) of the class of a given class-label.
These results are general, holding for any classification, that is, for any choice of $P_\text{II*}$.

\subsection{The structure of the Level~I-II correlations}
\label{sec:struct.corr}

In the Level I classification of correlations,
for the partitions $\xi,\xi'\in P_\text{I}$, we have
\begin{equation}
\label{eq:corrlatticeI.meet}
\mathcal{D}_{\xi \text{-unc}} \cap \mathcal{D}_{\xi' \text{-unc}}
= \mathcal{D}_{(\xi\wedge\xi')\text{-unc}}.
\end{equation}
This can be proven in the same way as
the same result for pure states was proven in
Appendix A.5 in~\cite{Szalay-2015b}.
Note that, due to the convex hull construction~\eref{eq:DsepI},
a similar identity does not hold in the Level I classification of entanglement.
($\wedge$ and $\vee$ are the greatest lower bound, or meet, and least upper bound, or join, 
in the respective lattice, see in~\ref{appsec:posets.defs}.)

In the Level II classification of correlations,
for the ideals $\vs{\xi},\vs{\xi}'\in P_\text{II}$, we have
\begin{equation}
\label{eq:corrlatticeII.join}
\mathcal{D}_{\vs{\xi} \text{-unc}} \cup \mathcal{D}_{\vs{\xi}' \text{-unc}}
= \mathcal{D}_{(\vs{\xi}\vee\vs{\xi}')\text{-unc}},
\end{equation}
and
\begin{equation}
\label{eq:corrlatticeII.meet}
\mathcal{D}_{\vs{\xi} \text{-unc}} \cap \mathcal{D}_{\vs{\xi}' \text{-unc}}
= \mathcal{D}_{(\vs{\xi}\wedge\vs{\xi}')\text{-unc}}.
\end{equation}
These can be proven
in the same way as
the same result for pure states was proven in
Appendix A.9 in~\cite{Szalay-2015b}
(\eref{eq:corrlatticeII.meet} relies also on~\eref{eq:corrlatticeI.meet}).
Note that, due to the convex hull construction~\eref{eq:DsepII},
similar identities do not hold in the Level II classification of entanglement.

\subsection{The structure of the correlation classes: existence}
\label{sec:struct.corrclassex}

A filter $\vvs{\xi}\in P_\text{III*}$ may lead to empty partial correlation class~\eref{eq:Cunc}.
Here we give necessary and sufficient condition for the labelling of the nonempty partial
correlation classes.

\begin{prop}
\label{prop:exist}
For a filter $\vvs{\xi}\in P_\text{III*}$,
the class $\mathcal{C}_{\vvs{\xi}\text{-unc}}\neq\emptyset$ 
if and only if $\wedge\vvs{\xi}\not\preceq\vee\cmpl{\vvs{\xi}}$.
\end{prop}
(We use the notations $\wedge\vvs{\xi}:=\bigwedge_{\vs{\xi}\in\vvs{\xi}}\vs{\xi}$ and
$\vee\cmpl{\vvs{\xi}}:=\bigvee_{\vs{\xi}'\in\cmpl{\vvs{\xi}}}\vs{\xi}'$.)

\begin{proof}
First, for a filter $\vvs{\xi}\in P_\text{III*}$, we write~\eref{eq:Cunc} as
\begin{equation*}
\mathcal{C}_{\vvs{\xi}\text{-unc}}
= \bigcap_{\vs{\xi}'\in\cmpl{\vvs{\xi}}} \cmpl{\mathcal{D}_{\vs{\xi}'\text{-unc}}} \cap 
   \bigcap_{\vs{\xi} \in      \vvs{\xi} }       \mathcal{D}_{\vs{\xi} \text{-unc}}
= \cmpl{\bigcup_{\vs{\xi}'\in\cmpl{\vvs{\xi}}} \mathcal{D}_{\vs{\xi}'\text{-unc}}} \cap 
   \bigcap_{\vs{\xi} \in      \vvs{\xi} }       \mathcal{D}_{\vs{\xi} \text{-unc}},
\end{equation*}
where the second equality is
De Morgan's law, 
then, applying~\eref{eq:corrlatticeII.join} and~\eref{eq:corrlatticeII.meet},
\begin{equation}
\label{eq:easy}
\mathcal{C}_{\vvs{\xi}\text{-unc}}
= \cmpl{\mathcal{D}_{(\vee\cmpl{\vvs{\xi}})\text{-unc}} } \cap
         \mathcal{D}_{(\wedge\vvs{\xi})\text{-unc}}.
\end{equation}
(Note that $\wedge\vvs{\xi},\vee\cmpl{\vvs{\xi}}\in P_\text{II}$ in general, 
they are not necessarily contained in $P_\text{II*}$,
since $P_\text{II*}$ is not necessarily a lattice.)
Now, since $B\subseteq A\inlineiff \cmpl{A}\cap B =\emptyset$,
we have that 
\begin{equation*}
\mathcal{C}_{\vvs{\xi}\text{-unc}} = \emptyset
\displayiff
\mathcal{D}_{(\wedge\vvs{\xi})\text{-unc}} \subseteq \mathcal{D}_{(\vee\cmpl{\vvs{\xi}})\text{-unc}},
\end{equation*}
which, applying~\eref{eq:DincII}, leads to
\begin{equation*}
\mathcal{C}_{\vvs{\xi}\text{-unc}} = \emptyset
\displayiff
\wedge\vvs{\xi} \preceq \vee\cmpl{\vvs{\xi}},
\end{equation*}
the contraposition of which is just Proposition~\ref{prop:exist}. 
\end{proof}

Note that a given filter $\vvs{\xi}\in P_\text{III*}$ may lead to empty or nonempty class,
depending on the choice of $P_\text{II*}\subseteq P_\text{II}$,
since, in the condition given in Proposition~\ref{prop:exist},
the complement $\cmpl{\vvs{\xi}}$ is given with respect to $P_\text{II*}$.
The fulfilment of the nonemptiness condition $\wedge\vvs{\xi}\not\preceq\vee\cmpl{\vvs{\xi}}$
is hard to check \emph{in general},
that is, without examining each $\vvs{\xi}\in P_\text{III*}$ one by one.
Now we give some tools which can be used for this,
and also for presenting general conditions 
for some important classifications $P_\text{II*}$, given in the subsequent sections.

\begin{lem}
\label{lem:equiv}
The following properties of a filter $\vvs{\xi}\in P_\text{III*}$ are equivalent:\\
(i) $\forall \vs{\xi}'\in P_\text{II*}$:
if $\vs{\xi}'\in\cmpl{\vvs{\xi}}$ then $\wedge\vvs{\xi} \not\preceq\vs{\xi}'$,\\
(i') $\forall \vs{\xi}\in P_\text{II*}$:
if $\wedge\vvs{\xi}\preceq\vs{\xi}$ then $\vs{\xi}\in\vvs{\xi}$,\\
(ii) $\upset\{\wedge\vvs{\xi}\}\cap P_\text{II*} = \vvs{\xi}$.
\end{lem}
\begin{proof} The steps are the following:\\
(i)~$\Leftrightarrow$~(i'): they are the contrapositions of each other.\\
(i')~$\Rightarrow$~(ii):
for $\vs{\xi}\in P_\text{II*}$,
$\wedge\vvs{\xi}\preceq\vs{\xi}$ means that
$\vs{\xi}\in\upset\{\wedge\vvs{\xi}\}\cap P_\text{II*}$,
that is,
by supposing (i'),
we have $\upset\{\wedge\vvs{\xi}\}\cap P_\text{II*}\subseteq\vvs{\xi}$.
The opposite inclusion holds in general:
for all $\vs{\xi}\in\vvs{\xi}$, we have that $\wedge\vvs{\xi}\preceq\vs{\xi}$,
since the meet $\wedge\vvs{\xi}$ is the greatest \emph{lower} bound of the elements of $\vvs{\xi}$,
so, because we also have $\vs{\xi}\in P_\text{II*}$, we end up with
$\vs{\xi}\in\upset\{\wedge\vvs{\xi}\}\cap P_\text{II*}$,
that is, $\upset\{\wedge\vvs{\xi}\}\cap P_\text{II*}\supseteq\vvs{\xi}$.\\
(ii)~$\Rightarrow$~(i'):
all $\vs{\xi}\in P_\text{II*}$ such that $\wedge\vvs{\xi}\preceq\vs{\xi}$
is contained in $\upset\{\wedge\vvs{\xi}\}\cap P_\text{II*}$,
which equals to $\vvs{\xi}$ by the assumption, leading to $\vs{\xi}\in\vvs{\xi}$.
\end{proof}

\begin{lem}
\label{lem:suff}
For a filter $\vvs{\xi}\in P_\text{III*}$,
we have that
if $\wedge\vvs{\xi}\not\preceq\vee\cmpl{\vvs{\xi}}$,
then $\upset\{\wedge\vvs{\xi}\}\cap P_\text{II*} = \vvs{\xi}$.
\end{lem}

\begin{proof}
This can be proven contrapositively:
if $\upset\{\wedge\vvs{\xi}\}\cap P_\text{II*} \neq \vvs{\xi}$ then $\wedge\vvs{\xi}\preceq\vee\cmpl{\vvs{\xi}}$.
In Lemma~\ref{lem:equiv}, (ii) does not hold if and only if (i) does not hold,
which means that there exists $\vs{\xi}'\in \cmpl{\vvs{\xi}}$ which is
$\wedge\vvs{\xi}\preceq\vs{\xi}'$.
With this $\vs{\xi}'$ we have $\wedge\vvs{\xi}\preceq\vs{\xi}'\preceq\vee\cmpl{\vvs{\xi}}$,
leading to $\wedge\vvs{\xi}\preceq\vee\cmpl{\vvs{\xi}}$
by the transitivity of the partial order.
\end{proof}

\begin{lem}
\label{lem:suffd}
For a filter $\vvs{\xi}\in P_\text{III*}$,
we have that
if $\wedge\vvs{\xi}\not\preceq\vee\cmpl{\vvs{\xi}}$,
then $\downset\{\vee\cmpl{\vvs{\xi}}\}\cap P_\text{II*} = \cmpl{\vvs{\xi}}$.
\end{lem}
\begin{proof}
This is the dual of Lemma~\ref{lem:suff}, so it can be proven analogously
(by proving also the dual of Lemma~\ref{lem:equiv}).
\end{proof}

Lemma~\ref{lem:suff} and Lemma~\ref{lem:suffd} tell us that
the nonempty classes can be labelled by \emph{principal filters} restricted to $P_\text{II*}$.
The reverse is not true in general.

\begin{lem}
\label{lem:mm}
For a filter $\vvs{\xi}\in P_\text{III*}$,
we have that
$\wedge\vvs{\xi}\not\preceq\vee\cmpl{\vvs{\xi}}$
if and only if 
$\upset\{\wedge\vvs{\xi}\}\cap\downset\{\vee\cmpl{\vvs{\xi}}\}=\emptyset$.
\end{lem}
 
\begin{proof}
This is 
the special case of
the contraposition of
that,
for all $\vs{\upsilon},\vs{\upsilon}'\in P_\text{II}$,
we have that
$\vs{\upsilon}\preceq\vs{\upsilon}'$
if and only if
$\upset\{\vs{\upsilon}\}\cap\downset\{\vs{\upsilon}'\}\neq\emptyset$.\\
To see the \textit{``if''} implication,
we have an $\vs{\zeta}\in \upset\{\vs{\upsilon}\}$,
which is also
$\vs{\zeta}\in \downset\{\vs{\upsilon}'\}$,
that is,
$\vs{\upsilon}\preceq\vs{\zeta}$ and
$\vs{\zeta}\preceq\vs{\upsilon}'$,
leading to that $\vs{\upsilon}\preceq\vs{\upsilon}'$
by the transitivity of the partial order.\\
To see the \textit{``only if''} implication,
we have that 
$\vs{\upsilon}\in\upset\{\vs{\upsilon}\}$ obviously, and
$\vs{\upsilon}\in\downset\{\vs{\upsilon}'\}$ by the assumption,
so $\vs{\upsilon}\in\upset\{\vs{\upsilon}\}\cap\downset\{\vs{\upsilon}'\}$, which is then not empty.
\end{proof}

With the help of Lemma~\ref{lem:mm}, we can see the role of 
Lemma~\ref{lem:suff} and Lemma~\ref{lem:suffd}.
For a filter $\vvs{\xi}\in P_\text{III*}$, 
using Proposition~\ref{prop:exist} and Lemma~\ref{lem:mm}, we have \emph{in general} that
\begin{equation*}
\mathcal{C}_{\vvs{\xi}\text{-unc}}\neq\emptyset
\displayiff
\wedge\vvs{\xi} \not\preceq \vee\cmpl{\vvs{\xi}}
\displayiff
\upset\{\wedge\vvs{\xi}\}\cap\downset\{\vee\cmpl{\vvs{\xi}}\}=\emptyset
\end{equation*}
\begin{equation*}
\qquad\displaythen
\bigl(\upset\{\wedge\vvs{\xi}\}\cap P_\text{II*}\bigr) \cap
\bigl(\downset\{\vee\cmpl{\vvs{\xi}}\}\cap P_\text{II*}\bigr) =\emptyset
\displaythen
\vvs{\xi} \cap \cmpl{\vvs{\xi}} =\emptyset,
\end{equation*}
where we have used at the last arrow that 
$\vvs{\xi} \subseteq \upset\{\wedge\vvs{\xi}\}\cap P_\text{II*}$
and
$\cmpl{\vvs{\xi}}\subseteq\downset\{\vee\cmpl{\vvs{\xi}}\}\cap P_\text{II*}$,
which hold in general (see in the (i')~$\Rightarrow$~(ii) implication of the proof of Lemma~\ref{lem:equiv}).
Note that Lemma~\ref{lem:suff} and Lemma~\ref{lem:suffd} \emph{tell more:}
if $\mathcal{C}_{\vvs{\xi}\text{-unc}}\neq\emptyset$, then 
$\vvs{\xi} = \upset\{\wedge\vvs{\xi}\}\cap P_\text{II*}$
and
$\cmpl{\vvs{\xi}} = \downset\{\vee\cmpl{\vvs{\xi}}\}\cap P_\text{II*}$
in the above conditions. 
So the last arrow is $\Longleftrightarrow$,
if we restrict to the nonempty case.

Note, on the other hand, that Lemma~\ref{lem:suff} and Lemma~\ref{lem:suffd} tell us 
that, for \emph{nonempty} classes,
$\vvs{\xi}$, $\wedge\vvs{\xi}$ and $\vee\cmpl{\vvs{\xi}}$ determine one another.
For example,
if $\wedge\vvs{\xi}=\wedge\vvs{\upsilon}$
then $\upset\{\wedge\vvs{\xi}\}=\upset\{\wedge\vvs{\upsilon}\}$, 
then $\upset\{\wedge\vvs{\xi}\}\cap P_\text{II*}=\upset\{\wedge\vvs{\upsilon}\}\cap P_\text{II*}$,
then, by Lemma~\ref{lem:suff}, $\vvs{\xi}=\vvs{\upsilon}$,
while the reverse implication is obvious.

\subsection{The structure of the correlation classes: uniqueness}
\label{sec:struct.corrclassunique}

Two different filters $\vvs{\xi},\vvs{\upsilon}\in P_\text{III*}$
may lead to the same partial correlation class~\eref{eq:Cunc}.
Here we give necessary and sufficient condition for the unique labelling of the partial correlation classes.

\begin{prop}
\label{prop:unique}
For the filters $\vvs{\xi},\vvs{\upsilon}\in P_\text{III*}$,
the classes
$\mathcal{C}_{\vvs{\xi}\text{-unc}}=\mathcal{C}_{\vvs{\upsilon}\text{-unc}}$
if and only if the following conditions hold:
\begin{eqnarray*}
(\wedge\vvs{\xi})     \wedge (\vee\cmpl{\vvs{\upsilon}}) \preceq \vee\cmpl{\vvs{\xi}}, \qquad
\wedge\vvs{\xi}                                          \preceq (\vee\cmpl{\vvs{\xi}}) \vee (\wedge\vvs{\upsilon}),\\
(\wedge\vvs{\upsilon})\wedge (\vee\cmpl{\vvs{\xi}})      \preceq \vee\cmpl{\vvs{\upsilon}}, \qquad
\wedge\vvs{\upsilon}                                     \preceq (\vee\cmpl{\vvs{\upsilon}}) \vee (\wedge\vvs{\xi}).
\end{eqnarray*}
\end{prop}

\begin{proof}
This can be proven by standard set theory.\\
$\mathcal{C}_{\vvs{\xi}\text{-unc}}=\mathcal{C}_{\vvs{\upsilon}\text{-unc}}$
if and only if\\
$\mathcal{C}_{\vvs{\xi}\text{-unc}}\subseteq\mathcal{C}_{\vvs{\upsilon}\text{-unc}}$ and
$\mathcal{C}_{\vvs{\xi}\text{-unc}}\supseteq\mathcal{C}_{\vvs{\upsilon}\text{-unc}}$,
if and only if\\
$\mathcal{C}_{\vvs{\xi}\text{-unc}}\cap \cmpl{\mathcal{C}_{\vvs{\upsilon}\text{-unc}}}=\emptyset$ and
$\cmpl{\mathcal{C}_{\vvs{\xi}\text{-unc}}}\cap \mathcal{C}_{\vvs{\upsilon}\text{-unc}}=\emptyset$. \\
Using~\eref{eq:easy} (based on the definition~\eref{eq:Cunc}), De Morgan's law and the distributivity,
we end up with that the above is equivalent to \\
$\bigl(
\cmpl{\mathcal{D}_{\vee\cmpl{\vvs{\xi}}\text{-unc}}} \cap 
      \mathcal{D}_{\wedge\vvs{\xi}\text{-unc}} \cap
      \mathcal{D}_{\vee\cmpl{\vvs{\upsilon}}\text{-unc}} \bigr)\cup\bigl( 
\cmpl{\mathcal{D}_{\vee\cmpl{\vvs{\xi}}\text{-unc}}} \cap
\cmpl{\mathcal{D}_{\wedge\vvs{\upsilon}\text{-unc}}} \cap
      \mathcal{D}_{\wedge\vvs{\xi}\text{-unc}} \bigr)=\emptyset$
and \\
$\bigl(
\cmpl{\mathcal{D}_{\vee\cmpl{\vvs{\upsilon}}\text{-unc}}} \cap 
      \mathcal{D}_{\wedge\vvs{\upsilon}\text{-unc}} \cap
      \mathcal{D}_{\vee\cmpl{\vvs{\xi}}\text{-unc}} \bigr)\cup\bigl( 
\cmpl{\mathcal{D}_{\vee\cmpl{\vvs{\upsilon}}\text{-unc}}} \cap
\cmpl{\mathcal{D}_{\wedge\vvs{\xi}\text{-unc}}} \cap
      \mathcal{D}_{\wedge\vvs{\upsilon}\text{-unc}} \bigr)=\emptyset$.\\
Using De Morgan's law,~\eref{eq:corrlatticeII.join} and~\eref{eq:corrlatticeII.meet},
and that $A\cup B=\emptyset\inlineiff (A=\emptyset\;\text{and}\;B=\emptyset)$,
this holds if and only if \\
$\cmpl{\mathcal{D}_{\vee\cmpl{\vvs{\xi}}\text{-unc}}} \cap
\mathcal{D}_{(\wedge\vvs{\xi})\wedge(\vee\cmpl{\vvs{\upsilon}})\text{-unc}}=\emptyset$ and
$\cmpl{ \mathcal{D}_{(\vee\cmpl{\vvs{\xi}})\vee(\wedge\vvs{\upsilon})\text{-unc}}} \cap
\mathcal{D}_{\wedge\vvs{\xi}\text{-unc}}=\emptyset$ and \\
$\cmpl{\mathcal{D}_{\vee\cmpl{\vvs{\upsilon}}\text{-unc}}} \cap
\mathcal{D}_{(\wedge\vvs{\upsilon})\wedge(\vee\cmpl{\vvs{\xi}})\text{-unc}}=\emptyset$ and
$\cmpl{\mathcal{D}_{(\vee\cmpl{\vvs{\upsilon}})\vee(\wedge\vvs{\xi})\text{-unc}}} \cap
\mathcal{D}_{\wedge\vvs{\upsilon}\text{-unc}}=\emptyset$.\\
Now, after using that 
$B\subseteq A\inlineiff \cmpl{A}\cap B =\emptyset$, \eref{eq:DincII} completes the proof.
\end{proof}

Note that the conditions in Proposition~\ref{prop:unique}
are weaker than the emptiness conditions 
$\wedge\vvs{\xi}\preceq \vee\cmpl{\vvs{\xi}}$ and 
$\wedge\vvs{\upsilon}\preceq \vee\cmpl{\vvs{\upsilon}}$ 
by Proposition~\ref{prop:exist},
and express the interrelation of $\vvs{\xi}$ and $\vvs{\upsilon}$.

\section{The structure of the correlation classes: examples}
\label{sec:corrclassxmpl}

In this section,
applying the results of the previous section,
we elaborate the structure of the classification for some important choices of $P_\text{II*}$,
namely, 
for the \emph{finest classification} (section~\ref{sec:corrclassxmpl.finest}),
for \emph{chain-based classifications} (section~\ref{sec:corrclassxmpl.chkPP}), 
 specially for \emph{$k$-partitionability} and \emph{$k$-producibility classifications,}
and for the classification based on the \emph{atoms of the correlation properties} (section~\ref{sec:corrclassxmpl.ach}).

\subsection{Finest classification}
\label{sec:corrclassxmpl.finest}

First, consider the finest classification, when $P_\text{II*} = P_\text{II}$.
We show that the structure of the correlation classes is isomorphic to the dual of $P_\text{I}$.

\begin{lem}
\label{lem:Cfines0}
Let $P_\text{II*} = P_\text{II}$, then, 
for a filter $\vvs{\xi}\in P_\text{III*}$,
the class $\mathcal{C}_{\vvs{\xi}\text{-unc}}\neq\emptyset$
if and only if
$\vvs{\xi}=\upset\{\wedge\vvs{\xi}\}$ and
$\cmpl{\vvs{\xi}}=\downset\{\vee\cmpl{\vvs{\xi}}\}$.
\end{lem}

\begin{proof}
To see the \textit{``only if''} implication, 
$\vvs{\xi}=\upset\{\wedge\vvs{\xi}\}\cap P_\text{II*}=\upset\{\wedge\vvs{\xi}\}$ and
$\cmpl{\vvs{\xi}}=\downset\{\vee\cmpl{\vvs{\xi}}\}\cap P_\text{II*}=\downset\{\vee\cmpl{\vvs{\xi}}\}$
by Proposition~\ref{prop:exist}, Lemma~\ref{lem:suff} and Lemma~\ref{lem:suffd}.\\
To see the \textit{``if''} implication,
we have $\upset\{\wedge\vvs{\xi}\}\cap\downset\{\vee\cmpl{\vvs{\xi}}\}=\vvs{\xi}\cap\cmpl{\vvs{\xi}}=\emptyset$,
then Lemma~\ref{lem:mm} and Proposition~\ref{prop:exist} lead to the claim.
\end{proof}

\begin{prop}
\label{prop:Cfines}
Let $P_\text{II*} = P_\text{II}$, then, 
for a filter $\vvs{\xi}\in P_\text{III*}$,
the class $\mathcal{C}_{\vvs{\xi}\text{-unc}}\neq\emptyset$
if and only if $\exists\xi\in P_\text{I}$ such that $\vvs{\xi}=\upset\{\downset\{\xi\}\}$.
\end{prop}

\begin{proof}
Proposition~\ref{prop:Cfines} can be reformulated by Lemma~\ref{lem:Cfines0}:
$\vvs{\xi}=\upset\{\wedge\vvs{\xi}\}$ and
$\cmpl{\vvs{\xi}}=\downset\{\vee\cmpl{\vvs{\xi}}\}$
if and only if
$\exists\xi\in P_\text{I}$ such that $\vvs{\xi}=\upset\{\downset\{\xi\}\}$.
This can be proven as follows.\\
To see the \textit{``if''} implication,
on the one hand,
we have that if $\vvs{\xi}=\upset\{\downset\{\xi\}\}$ for a $\xi\in P_\text{I}$,
then $\wedge\vvs{\xi}=\downset\{\xi\}$, so $\vvs{\xi}=\upset\{\wedge\vvs{\xi}\}$.
On the other hand, 
$\vvs{\xi}=\upset\{\downset\{\xi\}\}
=\sset{\vs{\xi}\in P_\text{II}}{\vs{\xi}\succeq\downset\{\xi\}}
=\sset{\vs{\xi}\in P_\text{II}}{\xi\in\vs{\xi}}$,
so its complement (with respect to $P_\text{II*}=P_\text{II}$)
is $\cmpl{\vvs{\xi}}
=\sset{\vs{\xi}'\in P_\text{II}}{\xi\notin\vs{\xi}'}$,
and we claim that
$\vee\cmpl{\vvs{\xi}}
=\cmpl{\upset\{\xi\}}
\equiv\sset{\xi'\in P_\text{I}}{\xi\not\preceq\xi'}$.
To see the $\supseteq$ inclusion,
we have that $\forall \xi'\in P_\text{I}$ which is $\xi'\not\succeq\xi$,
for the down-set $\vs{\xi}':=\downset\{\xi'\}$ we have
$\xi\notin\vs{\xi}'$, 
so $\vs{\xi}'\in\cmpl{\vvs{\xi}}$,
so $\xi'\in\vs{\xi}'\preceq\vee\cmpl{\vvs{\xi}}$.
To see the $\subseteq$ inclusion,
we use contraposition.
For all $\xi'\in P_\text{I}$ such that $\xi'\succeq\xi$,
every $\vs{\xi}'\in P_\text{II*}$ for which $\xi'\in\vs{\xi}'$ we also have $\xi\in\vs{\xi}'$,
because $\vs{\xi}'$ is a down-set,
so $\vs{\xi}'\notin\cmpl{\vvs{\xi}}$.
Because this holds for all such $\vs{\xi}'$, we have
$\xi'\notin\vee\cmpl{\vvs{\xi}}$.
Now, we have  
$\vee\cmpl{\vvs{\xi}}=\cmpl{\upset\{\xi\}}$,
and we have to prove that
$\cmpl{\vvs{\xi}}=\downset\{\vee\cmpl{\vvs{\xi}}\}$.
By definition, and the results for $\cmpl{\vvs{\xi}}$ and $\vee\cmpl{\vvs{\xi}}$ above, 
we have to prove the third equality in 
$\downset\{\vee\cmpl{\vvs{\xi}}\}
=\downset\{\cmpl{\upset\{\xi\}}\}
=\sset{\vs{\xi}'\in P_\text{II}}{\vs{\xi}'\preceq \cmpl{\upset\{\xi\}}}
=\sset{\vs{\xi}'\in P_\text{II}}{\xi\notin\vs{\xi}'}
=\cmpl{\vvs{\xi}}$.
This can be seen as
$\vs{\xi}'\preceq \cmpl{\upset\{\xi\}}$
if and only if the up-sets
$\cmpl{\vs{\xi}'}\succeq \upset\{\xi\}$
if and only if 
$\xi\in\cmpl{\vs{\xi}'}$
if and only if 
$\xi\not\in\vs{\xi}'$.\\
To see the \textit{``only if''} implication,
we prove the contrapositive statement.
If $\vvs{\xi}\neq\upset\{\downset\{\xi\}\}$ for a $\xi\in P_\text{I}$,
then
we have two possibilities.
First, if $\vvs{\xi}\neq\upset\{\vs{\xi}\}$ for a $\vs{\xi}\in P_\text{II}$,
then $\wedge\vvs{\xi}\notin\vvs{\xi}$, then $\vvs{\xi}\neq \upset\{\wedge\vvs{\xi}\}$.
Second, 
although $\vvs{\xi}=\upset\{\vs{\xi}\}$,
we have $\vs{\xi}=\downset M\in P_\text{II}$, where
$M=\max(\vs{\xi})=\{\xi_1,\xi_2,\dots,\xi_m\}$ with $m\geq2$
(each down-set is the down-closure of its maximal elements).
In this case, although we have $\vvs{\xi}=\upset\{\vs{\xi}\}=\upset\{\wedge\vvs{\xi}\}$ by $\wedge\vvs{\xi}=\vs{\xi}$,
we will have $\downset\{\vee\cmpl{\vvs{\xi}}\}\neq\cmpl{\vvs{\xi}}$.
Indeed,
$\vvs{\xi}=\upset\{\wedge\vvs{\xi}\}=\upset\{\downset M\}
=\sset{\vs{\xi}\in P_\text{II}}{\downset M \preceq\vs{\xi}}
=\sset{\vs{\xi}\in P_\text{II}}{M \subseteq\vs{\xi}}$,
then
$\cmpl{\vvs{\xi}}=\sset{\vs{\xi}'\in P_\text{II}}{\exists \xi\in M:\xi\notin\vs{\xi}'}\not\ni \downset M$;
however, since $m\geq2$,
the union of such down-sets $\vs{\xi}'$ contains all $\xi\in M$,
that is, $M \subseteq\vee\cmpl{\vvs{\xi}}$,
so $\downset\{\vee\cmpl{\vvs{\xi}}\}=\sset{\vs{\xi}'\in P_\text{II}}{\vs{\xi}'\preceq\vee\cmpl{\vvs{\xi}}}\ni \downset M$,
leading to that $\downset\{\vee\cmpl{\vvs{\xi}}\}\neq\cmpl{\vvs{\xi}}$.
\end{proof}

\begin{prop}
\label{prop:Cfinesu}
Let $P_\text{II*} = P_\text{II}$, then,
for the partitions $\xi,\upsilon\in P_\text{I}$,
the classes $\mathcal{C}_{\upset\{\downset\{\xi\}\}\text{-unc}} = \mathcal{C}_{\upset\{\downset\{\upsilon\}\}\text{-unc}}$
if and only if $\xi=\upsilon$.
\end{prop}

\begin{proof}
The \textit{``if''} implication is obvious, to see the \textit{``only if''} implication,
we have in Proposition~\ref{prop:Cfines} that if 
$\vvs{\xi}=\upset\{\downset\{\xi\}\}$ and
$\vvs{\upsilon}=\upset\{\downset\{\upsilon\}\}$ 
for  $\xi,\upsilon \in P_\text{I}$,
then
$\wedge\vvs{\xi}=\downset\{\xi\}$, $\vee\cmpl{\vvs{\xi}}=\cmpl{\upset\{\xi\}}$,
$\wedge\vvs{\upsilon}=\downset\{\upsilon\}$, $\vee\cmpl{\vvs{\upsilon}}=\cmpl{\upset\{\upsilon\}}$,
which can be used in the conditions in Proposition~\ref{prop:unique}.
For example, the top-right one is then
$\downset\{\xi\}\preceq\cmpl{\upset\{\xi\}}\vee\downset\{\upsilon\}$,
which, 
since $\xi\in\downset\{\xi\}$,
tells us that
$\xi\in \cmpl{\upset\{\xi\}}\vee\downset\{\upsilon\}$.
Since $\xi\notin\cmpl{\upset\{\xi\}}$,
we have that $\xi\in\downset\{\upsilon\}$,
that is, $\xi\preceq\upsilon$.
It can be seen similarly (from, for example, the lower right condition in Proposition~\ref{prop:unique})
that $\upsilon\preceq\xi$, leading to that $\xi=\upsilon$.
\end{proof}

In summary, we have that
the nonempty classes can be labelled by the \emph{principal filters}
generated by the \emph{principal ideals} of partitions uniquely.
So, contrary to the same case of entanglement,
we could actually skip Level II in this case;
however, it is needed in the general construction, for example,
in $k$-partitionability and $k'$-producibility based classifications.
It also follows that the number of the classes is the same as the number of the possible partitions,
$\abs{P_\text{I}}$, given by the Bell numbers~\cite{oeisA000110}.
The \emph{strictly $\upset\{\downset\{\xi\}\}$-uncorrelated states} are those,
which are $\xi$-uncorrelated, while correlated with respect to any finer partition;
and no other label is meaningful.
For example, for $n=3$ we have the five classes
$\mathcal{C}_{\upset\{\downset\{1|2|3\}\}} = \{\varrho_{1}\otimes\varrho_{2}\otimes\varrho_{3}\}$,
$\mathcal{C}_{\upset\{\downset\{ab|c\}\}} = \{\varrho_{ab}\otimes\varrho_{c}\}$,
$\mathcal{C}_{\upset\{\downset\{123\}\}} = \{\varrho_{123}\}$,
where, contrary to \eref{eq:DuncI}, the density operators $\varrho_X\in\mathcal{D}_X$ are not of product form,
and the formula is given for all choices of $a,b,c\in L$, such that $ab|c$ is a partition of $L$.
(Note that here we use a simplified notation for the partitions and subsystems, e.g.,
$12|3 = \{\{1,2\},\{3\}\}$.)
It is important to note here,
how simple the finest classification of correlations is ($1+3+1$ classes),
compared to the finest classification of entanglement ($1+18+1$ classes) \cite{Szalay-2015b}.
For $n=4$ we have the fifteen classes
$\mathcal{C}_{\upset\{\downset\{1|2|3|4\}\}} = \{\varrho_{1}\otimes\varrho_{2}\otimes\varrho_{3}\otimes\varrho_{4}\}$,
$\mathcal{C}_{\upset\{\downset\{ab|c|d\}\}} = \{\varrho_{ab}\otimes\varrho_{c}\otimes\varrho_{d}\}$,
$\mathcal{C}_{\upset\{\downset\{ab|cd\}\}} = \{\varrho_{ab}\otimes\varrho_{cd}\}$,
$\mathcal{C}_{\upset\{\downset\{abc|d\}\}} = \{\varrho_{abc}\otimes\varrho_{d}\}$,
$\mathcal{C}_{\upset\{\downset\{1234\}\}} = \{\varrho_{1234}\}$.

\subsection{Chains, \texorpdfstring{$k$}{k}-partitionability and \texorpdfstring{$k$}{k}-producibility}
\label{sec:corrclassxmpl.chkPP}

Second, consider the case when the classification is based
on properties which can be ordered totally.
Let $P_\text{II*}$ be a chain, that is,
$P_\text{II*}=\sset{\vs{\xi}_i}{i,j=1,2,\dots,\abs{P_\text{II*}},\; \vs{\xi}_i\preceq\vs{\xi}_j \inlineiff i\leq j}$.
We show that the structure
of the correlation classes is isomorphic to the dual of $P_\text{II*}$, 
so it also forms a chain.

\begin{prop}
\label{prop:Cchain}
Let $P_\text{II*}$ be a chain, then
the class $\mathcal{C}_{\vvs{\xi}\text{-unc}}\neq\emptyset$ for all filters $\vvs{\xi}\in P_\text{III*}$.
\end{prop}

\begin{proof}
An up-set $\vvs{\xi}$ of a chain $P_\text{II*}$ have a unique minimal element, $\min\vvs{\xi}=\{\vs{\xi}_\text{min}\}$, 
and then $\vs{\xi}_\text{min}=\wedge\vvs{\xi}$;
on the other hand, 
$\cmpl{\vvs{\xi}}$, the complement of the up-set $\vvs{\xi}$ is a down-set, 
and, similarly, 
a down-set $\cmpl{\vvs{\xi}}$ of a chain $P_\text{II*}$ have a unique maximal element, $\max\cmpl{\vvs{\xi}}=\{\vs{\xi}_\text{max}'\}$, 
and then $\vs{\xi}_\text{max}'=\vee\cmpl{\vvs{\xi}}$.
We also have $\vee\cmpl{\vvs{\xi}}=\vs{\xi}_\text{max}'\prec \vs{\xi}_\text{min}=\wedge\vvs{\xi}$,
since all pairs of elements in a chain $P_\text{II*}$ can be compared,
and $\vs{\xi}_\text{max}'\not\succeq\vs{\xi}_\text{min}$,
since 
in the other case $\vs{\xi}_\text{max}'$ would be contained in $\vvs{\xi}$, being an up-set.
Now, if
$\vee\cmpl{\vvs{\xi}}\prec \wedge\vvs{\xi}$,
then
$\vee\cmpl{\vvs{\xi}}\not\succeq \wedge\vvs{\xi}$,
and Proposition~\ref{prop:exist} leads to the claim.
\end{proof}

\begin{prop}
\label{prop:Cchainu}
Let $P_\text{II*}$ be a chain, then,
for the filters $\vvs{\xi},\vvs{\upsilon}\in P_\text{III*}$,
the classes
$\mathcal{C}_{\vvs{\xi}\text{-unc}}=\mathcal{C}_{\vvs{\upsilon}\text{-unc}}$
if and only if $\vvs{\xi}=\vvs{\upsilon}$.
\end{prop}

\begin{proof}
The \textit{``if''} implication is obvious, to see the \textit{``only if''} implication,
we have in Proposition~\ref{prop:Cchain} that, using the same notation,
$\vee\cmpl{\vvs{\xi}}=\vs{\xi}_\text{max}'\prec \vs{\xi}_\text{min}=\wedge\vvs{\xi}$, and
$\vee\cmpl{\vvs{\upsilon}}=\vs{\upsilon}_\text{max}'\prec \vs{\upsilon}_\text{min}=\wedge\vvs{\upsilon}$,
which can be used in the conditions in Proposition~\ref{prop:unique}.
For example, the top-right one is then
$\vs{\xi}_\text{min}\preceq \vs{\xi}_\text{max}'\vee \vs{\upsilon}_\text{min}$,
where the right-hand side is $\min\{\vs{\xi}_\text{max}',\vs{\upsilon}_\text{min}\}$,
since every pair of elements in a chain can be ordered.
Since $\vs{\xi}_\text{max}'\prec \vs{\xi}_\text{min}$,
the one remaining possibility on the right-hand side is $\vs{\upsilon}_\text{min}$,
leading to $\vs{\xi}_\text{min}\preceq \vs{\upsilon}_\text{min}$.
It can be seen similarly that
$\vs{\xi}_\text{min}\succeq \vs{\upsilon}_\text{min}$, 
leading to that $\vs{\xi}_\text{min}= \vs{\upsilon}_\text{min}$,
then $\vvs{\xi}=\vvs{\upsilon}$.
\end{proof}

Note that if $P_\text{II*}$ is a chain, then its up-sets in $P_\text{III*}$ form also a chain.
Then, in summary, we have that
the nonempty classes can be labelled by
all the \emph{principal filters} restricted to $P_\text{II*}$ uniquely.
It also follows that the number of the classes is the same as the number of the elements of the properties taken into account,
$\abs{P_\text{II*}}$.
Special cases are the partitionability and producibility classifications, when 
$P_\text{II*} = P_{\text{II\,part}} :=  \sset{\vs{\mu}_k}{k=1,2,\dots,n}$ and
$P_\text{II*} = P_{\text{II\,prod}} :=  \sset{\vs{\nu}_{k'}}{k'=1,2,\dots,n}$,
leading to the classes
of \emph{strictly $k$-partitionably} and \emph{strictly $k'$-producibly uncorrelated} states,
$\mathcal{C}_{k\text{-part\,unc}} := \mathcal{C}_{\upset\{\vs{\mu}_k\}\text{-unc}}$ and
$\mathcal{C}_{k'\text{-prod\,unc}}:= \mathcal{C}_{\upset\{\vs{\nu}_{k'}\}\text{-unc}}$, respectively.
In these cases we always have $n$ classes,
the class of \emph{genuine correlated states} is the class of
strictly $1$-partitionably, or equivalently, strictly $n$-producibly uncorrelated states;
while the class of \emph{totally uncorrelated states} is the class of
strictly $n$-partitionably, or equivalently, strictly $1$-producibly uncorrelated states.
(In general, there is no one-to-one correspondence between the partitionability and producibility correlations.)
For example, for $n=3$ we have 
$\mathcal{C}_{1\text{-part\,unc}} = \mathcal{C}_{3\text{-prod\,unc}} = \{\varrho_{123}\}$,
$\mathcal{C}_{2\text{-part\,unc}} = \mathcal{C}_{2\text{-prod\,unc}} = \{\varrho_{ab}\otimes\varrho_{c}\}$,
$\mathcal{C}_{3\text{-part\,unc}} = \mathcal{C}_{1\text{-prod\,unc}} = \{\varrho_{1}\otimes\varrho_{2}\otimes\varrho_{3}\}$,
with the notations used before.
For $n=4$, the two chains are different,
$\mathcal{C}_{1\text{-part\,unc}} = \mathcal{C}_{4\text{-prod\,unc}} = \{\varrho_{1234}\}$,
$\mathcal{C}_{2\text{-part\,unc}} = \{\varrho_{ab}\otimes\varrho_{cd}, \varrho_{abc}\otimes\varrho_{d}\}$,
$\mathcal{C}_{3\text{-prod\,unc}} = \{\varrho_{abc}\otimes\varrho_{d}\}$,
$\mathcal{C}_{3\text{-part\,unc}} = \{\varrho_{ab}\otimes\varrho_{c}\otimes\varrho_{d}\}$, 
$\mathcal{C}_{2\text{-prod\,unc}} = \{\varrho_{ab}\otimes\varrho_{c}\otimes\varrho_{d},\varrho_{ab}\otimes\varrho_{cd}\}$,
$\mathcal{C}_{4\text{-part\,unc}} = \mathcal{C}_{1\text{-prod\,unc}} = \{\varrho_{1}\otimes\varrho_{2}\otimes\varrho_{3}\otimes\varrho_{4}\}$.

\subsection{An antichain}
\label{sec:corrclassxmpl.ach}

Third, consider the case when the classification is based
on properties which cannot be ordered.
Let $P_\text{II*}$ be an antichain, that is,
$P_\text{II*}=\sset{\vs{\xi}_i}{i,j=1,2,\dots,\abs{P_\text{II*}},\; 
\vs{\xi}_i\not\preceq\vs{\xi}_j \inlineiff i\neq j}$.
Then every subset of this is automatically an up-set, so $P_\text{III*}=2^{P_\text{II*}}\setminus\{\emptyset\}$.
One cannot formulate a general result in this case,
as was done for chains,
Proposition~\ref{prop:exist} and Proposition~\ref{prop:unique} have to be checked for the filters $\vvs{\xi}\in P_\text{III*}$.
For at least one particular antichain, the antichain of the atoms of the correlation properties,
however, we can obtain the complete classification.

\begin{prop}
\label{prop:Catoms}
Let 
$P_\text{II*}=\sset{\downset\{\xi\}}{\xi\in P_\text{I}, \abs{\xi}=n-1}$, then,
for a filter $\vvs{\xi}\in P_\text{III*}$,
the class $\mathcal{C}_{\vvs{\xi}\text{-unc}}\neq\emptyset$ if and only if $\abs{\vvs{\xi}}=1$.
\end{prop}

\begin{proof}
To see the \textit{``if''} implication,
$\abs{\vvs{\xi}}=1$ for a $\vvs{\xi}\in P_\text{III*}$
means that $\vvs{\xi} = \upset\{\downset\{\xi\}\}\cap P_\text{II*}$ 
 for a $\xi\in P_\text{I}$.
Then $\wedge\vvs{\xi} = \downset\{\xi\}$, so $\xi\in\wedge\vvs{\xi}$.
On the other hand,
$\cmpl{\vvs{\xi}}=\upset\sset{\downset\{\xi'\} \in P_\text{II*}}{\xi'\neq\xi}\cap P_\text{II*}$,
so $\xi\notin\vee\cmpl{\vvs{\xi}}$, 
since $\xi\notin\downset\{\xi'\}$ for all $\xi'\neq\xi$,
since $\downset\{\xi_i'\}=\{\xi_i',\bot\}$
(where $\bot\in P_\text{I}$ is the finest partition, the \emph{bottom element} of $P_\text{I}$).
So we have that $\wedge\vvs{\xi}\not\preceq\vee\cmpl{\vvs{\xi}}$, 
then Proposition~\ref{prop:exist} leads to that $\mathcal{C}_{\vvs{\xi}\text{-unc}}\neq\emptyset$.\\
To see the \textit{``only if''} implication, we prove the contrapositive statement.
Let $\abs{\vvs{\xi}}\geq2$ for a $\vvs{\xi}\in P_\text{III*}$,
that is, 
 for some distinct partitions $\xi_1,\xi_2,\dots,\xi_m\in P_\text{I}$,
we have
$\vvs{\xi} = \upset\{\downset\{\xi_1\},\downset\{\xi_2\},\dots,\downset\{\xi_m\}\}\cap P_\text{II*}$ for $m=\abs{\vvs{\xi}}\geq2$.
Since $\downset\{\xi_i\}=\{\xi_i,\bot\}$, 
we have that $\wedge\vvs{\xi}=\{\bot\}$.
Since $\{\bot\}$ is the bottom element of $P_\text{II}$,
we have $\wedge\vvs{\xi}\preceq\vee\cmpl{\vvs{\xi}}$, without the need for the calculation of $\cmpl{\vvs{\xi}}$,
then Proposition~\ref{prop:exist} leads to that $\mathcal{C}_{\vvs{\xi}\text{-unc}}=\emptyset$.
\end{proof}

\begin{prop}
\label{prop:Catomsu}
Let $P_\text{II*}=\sset{\downset\{\xi\}}{\xi\in P_\text{I}, \abs{\xi}=n-1}$, then,
for the partitions $\xi,\upsilon\in P_\text{I}$ with $\abs{\xi}=\abs{\upsilon}=n-1$, 
the classes $\mathcal{C}_{\upset\{\downset\{\xi\}\}\text{-unc}} = \mathcal{C}_{\upset\{\downset\{\upsilon\}\}\text{-unc}}$
if and only if $\xi=\upsilon$.
\end{prop}

\begin{proof}
The \textit{``if''} implication is obvious, to see the \textit{``only if''} implication,
we have in Proposition~\ref{prop:Catoms} that if
$\vvs{\xi}=\upset\{\downset\{\xi\}\}\cap P_\text{II*}$ and
$\vvs{\upsilon}=\upset\{\downset\{\upsilon\}\}\cap P_\text{II*}$
for $\xi,\upsilon \in P_\text{I}$,
then 
$\wedge\vvs{\xi} = \downset\{\xi\}\ni\xi$ and
$\wedge\vvs{\upsilon} = \downset\{\upsilon\}\ni\upsilon$, while
$\xi\notin\vee\cmpl{\vvs{\xi}}$ and
$\upsilon\notin\vee\cmpl{\vvs{\upsilon}}$,
which can be used in the conditions in Proposition~\ref{prop:unique}.
For example, the top-right one 
takes the form $\downset\{\xi\} \preceq (\vee\cmpl{\vvs{\xi}}) \vee (\downset\{\upsilon\})$,
so,
since $\xi\in\downset\{\xi\}$ and $\xi\notin\vee\cmpl{\vvs{\xi}}$,
we have that $\xi\in\wedge\vvs{\upsilon}=\downset\{\upsilon\} =\{\upsilon,\bot\}$, 
leading to that $\xi=\upsilon$.
\end{proof}

Note that the antichain $P_\text{II*}$ we considered here 
is the antichain of the \emph{atoms} of the lattice $P_\text{II}$,
being the \emph{principal ideals} generated by the \emph{$(n-1)$-partitions}, being the atoms of $P_\text{I}$.
In the $(n-1)$-partitions the only non-singlepartite subsystem is bipartite,
the correlations given by these partitions can be considered ``elementary'' in some sense.
Then, in summary, we have that
the nonempty classes can be labelled by 
the \emph{principal filters} restricted to $P_\text{II*}$
generated by the \emph{principal ideals} of $(n-1)$-partitions uniquely.
It also follows that the number of the classes is ${n\choose 2}$.
For example, for $n=3$ we have the three classes
$\mathcal{C}_{\upset\{\downset\{ab|c\}\}} = \{\varrho_{ab}\otimes\varrho_{c}\}$,
for $n=4$ we have the six classes
$\mathcal{C}_{\upset\{\downset\{ab|c|d\}\}} = \{\varrho_{ab}\otimes\varrho_{c}\otimes\varrho_{d}\}$,
with the notations used before.
This classification does not cover the whole state space.
More useful would be to consider the classification, analoguous to this by duality,
based on the antichain of the
principal ideals generated by the \emph{bipartitions}
($P_\text{II*}=\sset{\downset\{\xi\}}{\xi\in P_\text{I}, \abs{\xi}=2}$),
being the \emph{coatoms} of $P_\text{I}$.
This cannot be done simply by duality,
because we have to consider down-sets in both cases,
they cannot be replaced with up-sets, which are the dual notions.
In this case one has to check Proposition~\ref{prop:exist} and Proposition~\ref{prop:unique}
for all filters $\vvs{\xi}\in P_\text{III*}$ one by one. 
%

\section{Summary, remarks and open questions}
\label{sec:Summ}

In this work,
we have considered the \emph{partial correlation classification}~\eref{eq:Cunc},
and we have given
necessary and sufficient conditions for the existence (Proposition~\ref{prop:exist}) 
and uniqueness (Proposition~\ref{prop:unique}) of the class of a given class-label.
The importance of the results,
and the reason for using the robust machinery,
is that all the possible partial correlation based classifications
can be described in this general way.
Particular cases we considered were 
the finest classification,
the classification based on chains in general (including $k$-partitionability and $k$-producibility),
and the classification based on the atoms of the correlation properties,
in which cases we could formulate the classification in an explicit manner.

For the \emph{partial entanglement classification}~\eref{eq:Csep},
such results cannot be obtained.
The reason for this is that
the lattice isomorphism~\eref{eq:corrlatticeII.join}-\eref{eq:corrlatticeII.meet}, 
which holds for the partial correlation,
does not hold for partial entanglement, 
we have only~\cite{Szalay-2015b}
\begin{equation}
\mathcal{D}_{\vs{\xi} \text{-sep}} \cup \mathcal{D}_{\vs{\xi}' \text{-sep}}
\subseteq \mathcal{D}_{(\vs{\xi}\vee\vs{\xi}')\text{-sep}}
\end{equation}
and
\begin{equation}
\mathcal{D}_{\vs{\xi} \text{-sep}} \cap \mathcal{D}_{\vs{\xi}' \text{-sep}}
\supseteq \mathcal{D}_{(\vs{\xi}\wedge\vs{\xi}')\text{-sep}}.
\end{equation}
It is still a conjecture that $\mathcal{C}_{\vvs{\xi}\text{-sep}}$ is nonempty
and unique for all $\vvs{\xi}\in P_\text{III*}$~\cite{Szalay-2015b}.
Note, however, that entanglement in \emph{pure states} is simply the correlation, so
our present results can be applied for the partial entanglement classification of pure states.

Note that,
although Level~II of the construction
is originally motivated by the need for the description of statistical mixtures of different product states~\eref{eq:DsepII}
in multipartite \emph{entanglement} theory~\cite{Szalay-2015b},
it is also meaningful when multipartite \emph{correlations} are considered~\cite{Szalay-2017}
(without mixtures~\eref{eq:DuncII}).
In the latter case, it describes the different \emph{possibilities} for productness:
taking the union of state spaces~\eref{eq:DuncII} expresses \emph{logical disjunction},
so using Level~II makes possible to
handle correlation and entanglement properties in an overall sense, without respect to a specific partition. 
This is why we identify Level~II as encoding
the \emph{aspects} or \emph{properties}
of \emph{partial correlation and entanglement.}

We mention that the corresponding (information-geometry based) 
\emph{correlation} and \emph{entanglement measures} are given 
for all $\xi$-correlation and $\xi$-entanglement (Level~I), and
for all $\vs{\xi}$-correlation and $\vs{\xi}$-entanglement (Level~II),
specially, for all $k$-partitionability and $k'$-producibility correlation and entanglement~\cite{Szalay-2015b,Szalay-2017}.
In a nutshell, these are the most natural generalizations of
the \emph{mutual information}~\cite{Petz-2008,Wilde-2013}, 
the \emph{entanglement entropy}~\cite{Bennett-1996a} and 
the \emph{entanglement of formation}~\cite{Bennett-1996b}
for the multipartite setting.
These are strong LO and LOCC monotones,
moreover, 
they show the same lattice structure as the partitions on Level I, $P_\text{I}$, 
and the partition ideals on Level II, $P_\text{II}$,
which is called \emph{multipartite monotonicity}~\cite{Szalay-2015b}.
For examples on the multipartite correlation measures,
evaluated for ground states of molecules, see \cite{Szalay-2017}.

\vspace{-8pt}
\ack
\vspace{-8pt}
Discussions with \emph{Mih\'aly M\'at\'e} are gratefully acknowledged.
This research was financially supported by
the {National Research, Development and Innovation Fund of Hungary}
within the \textit{Researcher-initiated Research Program} (project Nr:~NKFIH-K120569)
and within the \textit{Quantum Technology National Excellence Program} (project Nr:~2017-1.2.1-NKP-2017-00001),
and the {Hungarian Academy of Sciences}
within the \textit{J\'anos Bolyai Research Scholarship},
within the \textit{``Lend\"ulet'' Program}
and within the Czech-Hungarian Bilateral Mobility Grant (project no: P2015-102).

\section*{References}
\bibliographystyle{unsrt.bst}
\bibliography{CEclass}{}

\appendix
\section{Partially ordered sets}
\label{appsec:posets}

\begin{figure}\centering
\includegraphics{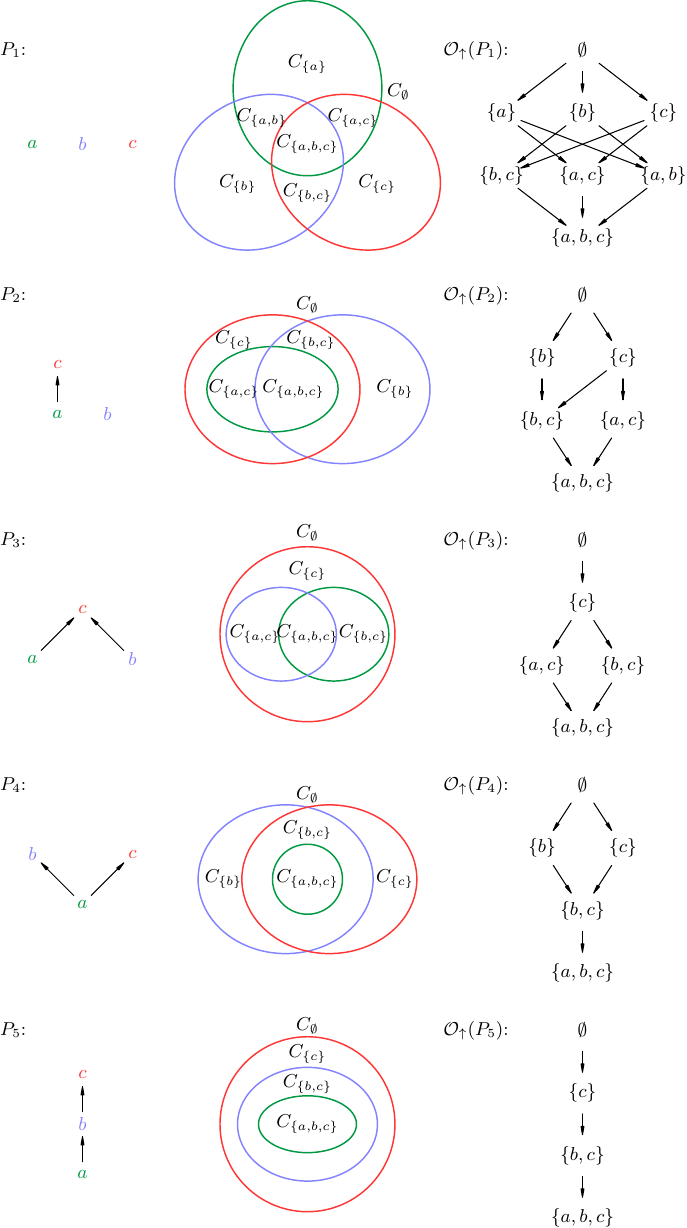}
\caption{
The poset of labels, $P$,
the Venn diagram of the inclusion of sets $A_x$ for the generic case,
and the lattice of the labels of the intersections (classes) being not empty by construction, $\mathcal{O}_\uparrow(P)$,
are shown
for the five possible posets (up to permutation) of three labels $P=\{a,b,c\}$~\cite{oeisA000112}.
}\label{fig:classVenn}
\end{figure}

\subsection{Elements in order theory}
\label{appsec:posets.defs}

Here we recall some elements in \emph{order theory}~\cite{Davey-2002,Roman-2008},
which are used in the main text.

A \emph{partially ordered set}, or \emph{poset}, $\struct{P,\preceq}$
is a set $P$ endowed with a \emph{partial order} $\preceq$,
which is a binary relation being
\emph{reflexive} ($\forall x\in P$: $x\preceq x$), 
\emph{antisymmetric} ($\forall x,y\in P$: if $x\preceq y$ and $y\preceq x$ then $y=x$)
and 
\emph{transitive} ($\forall x,y,z\in P$: if $x\preceq y$ and $y\preceq z$ then $x\preceq z$).
We consider finite posets ($\abs{P}<\infty$) only.
If every pair of elements can be related by $\preceq$, then the partial order is a total order,
and the poset is called a \emph{chain}.
If no pair of distinct elements can be related by $\preceq$, then the partial order is trivial,
and the poset is called an \emph{antichain}.

A poset $P$ may have a \emph{bottom element}, $\bot\in P$, and a \emph{top element}, $\top\in P$,
if $\forall x\in P$: $\bot\preceq x$, and $x\preceq\top$, respectively.
(If they exist, then they are unique, because of the antisymmetry of the ordering.)
If a (finite) poset $P$ has a bottom element,
then its \emph{atoms} are those $x$ elements for which 
$\forall y\in P$ if $y\prec x$ then $y=\bot$;
if a (finite) poset $P$ has a top element,
then its \emph{co-atoms} are those $x$ elements for which
$\forall y\in P$ if $x\prec y$ then $y=\top$.

The \emph{minimal} and \emph{maximal} elements of a subset $Q\subseteq P$ are
$\min Q = \sset{x\in Q}{(y\in Q\;\text{and}\;y\preceq x)\;\Rightarrow\;y=x}$, 
$\max Q = \sset{x\in Q}{(y\in Q\;\text{and}\;x\preceq y)\;\Rightarrow\;y=x}$.

A \emph{down-set}, or \emph{order ideal}, is a subset $Q\subseteq P$,
which is ``closed downwards'':
if $x\in Q$ and $y\preceq x$ then $y\in Q$.
An \emph{up-set}, or \emph{order filter}, is a subset $Q\subseteq P$,
which is ``closed upwards'':
if $x\in Q$ and $x\preceq y$ then $y\in Q$.
The sets of all down-sets and up-sets of $P$ 
are denoted with $\mathcal{O}_\downarrow(P)$ and $\mathcal{O}_\uparrow(P)$, respectively.

The \emph{down closure} and the \emph{up closure} of a subset $Q\subseteq P$ are
$\downset Q = \sset{x\in P}{\exists y\in Q: x\preceq y}$,
$\upset   Q = \sset{x\in P}{\exists y\in Q: y\preceq x}$,
which are a down-set (ideal) and an up-set (filter), respectively.
If $Q$ is a singleton, $\{x\}$,
then its down and up closures, $\downset\{x\}$ and $\upset\{x\}$,
are called \emph{principal ideal} and \emph{principal filter}, respectively. 

Elements $x,y\in P$ may have
\emph{greatest lower bound}, or \emph{meet}, $x\wedge y$
($x\wedge y \preceq x,y$, and $\forall z\in P$ if $z \preceq x,y$ then $z \preceq x\wedge y$)
and 
\emph{least upper bound}, or \emph{join}, $x\vee y$
($x,y \preceq x\vee y$, and $\forall z\in P$ if $x,y \preceq z$ then $x\vee y \preceq z$).
A (finite) poset $P$ is called a \emph{lattice},
if there exist meet and join for all pairs of its elements.
A (finite) lattice always has bottom and top elements.
Note that in the main text we use \emph{order ideals and filters,} which are just the down- and up-sets.
In the cases when the posets are lattices, lattice ideals and filters are considered automatically in the literature~\cite{Roman-2008}.
(\emph{Lattice ideals and filters} are nonempty down- and up-sets which inherit (finite) joins and meets.)
However, in our case, even when the posets considered are lattices,
our construction always uses order ideals and filters.

If we consider a power set,
the natural partial order $\preceq$ is the set inclusion $\subseteq$, 
then the meet $\wedge$ is the intersection $\cap$,
and the join $\vee$ is the union $\cup$.
For a poset $P$, $\mathcal{O}_\downarrow(P)$ and $\mathcal{O}_\uparrow(P)$ are lattices with respect to the inclusion.
If $P$ is a lattice, then also
$\mathcal{O}_\downarrow(P)\setminus\{\emptyset\}$ and 
$\mathcal{O}_\uparrow(P)\setminus\{\emptyset\}$ are lattices with respect to the inclusion.

\subsection{Intersections of sets}
\label{appsec:posets.int}

Let us have a set $A$,
and a finite number of its (different) subsets $A_x\in 2^A$,
labelled by elements $x\in P$ in a label set $P$.
All the possible intersections of the sets $A_x$ can be 
labelled by a subset $\vs{x}\in 2^P$ as
\begin{equation}
\label{aeq:CA}
C_{\vs{x}} := \bigcap_{x'\in\cmpl{\vs{x}}} \cmpl{A_{x'}}\cap\bigcap_{x\in\vs{x}}A_x \in 2^A,
\end{equation}
where the complement of the subset $A_x$ is in $2^A$, that is, $\cmpl{A_{x'}} = A\setminus A_{x'}$,
while the complement of the subset $\vs{x}$ is in $2^P$, that is, $\cmpl{\vs{x}} = P\setminus \vs{x}$.
(We use the convention that the empty intersection is the whole set $A$,
while the empty union is the empty set $\emptyset$.
Note that, for $\cmpl{A_{x'}}$, there does not necessarily exist $x\in P$ such that $A_x=\cmpl{A_{x'}}$.)

We would like to exploit the possible inclusions of the subsets $A_x$
in the labelling of the intersections.
In order to do this,
we endow the set $P$ of the labels with a partial order, based on the inclusion of the subsets $A_x$:
\begin{equation}
\label{aeq:Ainc}
\forall y,x\in P: \qquad y\preceq x \displayiff A_y\subseteq A_x.
\end{equation}

\begin{lem}
\label{lem:classint}
In the above setting, we have that
if $C_{\vs{x}} \neq \emptyset$ then $\vs{x}\in\mathcal{O}_\uparrow(P)$.
\end{lem}

\begin{proof}
This can be proven contrapositively:
If $\vs{x}$ is not an up-set ($\vs{x}\notin\mathcal{O}_\uparrow(P)$),
then there exists a pair of elements
$x\in\vs{x}$ and $x'\in \cmpl{\vs{x}}$ such that $x\preceq x'$,
then $A_x\subseteq A_{x'}$ by~\eref{aeq:Ainc},
then $\cmpl{A_{x'}}\cap A_x=\emptyset$,
then $C_{\vs{x}} = \emptyset$ by~\eref{aeq:CA}.
\end{proof}

\begin{lem}
\label{lem:classint2}
In the above setting, 
when $\bigcup_{x\in P}A_x=A$,
we have that
if $C_{\vs{x}} \neq \emptyset$ then $\vs{x}\in\mathcal{O}_\uparrow(P)\setminus\{\emptyset\}$.
\end{lem}

\begin{proof}
This is Lemma~\ref{lem:classint}
together with that 
in the case of the stronger assumption we have that
if $C_{\vs{x}} \neq \emptyset$ then $\vs{x}\neq\emptyset$.
This, again, can be proven contrapositively:
Let $\vs{x}=\emptyset$, then 
$C_\emptyset = \bigcap_{x'\in P} \cmpl{A_{x'}}
= \cmpl{\bigcup_{x'\in P} A_{x'} } = \cmpl{A}=\emptyset$,
where~\eref{aeq:CA} and De Morgan's law were used.
\end{proof}

Examples can be seen in figure~\ref{fig:classVenn}
(note that the up-set lattice $\mathcal{O}_\uparrow(P)$ is drawn upside-down, 
which is intuitive in the case of correlation and entanglement theory).
Lemma~\ref{lem:classint} and Lemma~\ref{lem:classint2}
give \emph{necessary} condition for the nonemptiness of the classes.
(If the condition does not hold, then the class can be called \emph{empty by construction}~\cite{Szalay-2015b}.)
It is \emph{not sufficient}, as one can see, for example, in figure~\ref{fig:classVennX}:
in the case when $P$ is an anti-chain,
it is possible that 
$A_a\not\subseteq A_b$, $A_a\not\subseteq A_c$,
while $A_a\subseteq A_b\cup A_c$, 
leading to $C_{\{a\}}=\cmpl{A_b}\cap\cmpl{A_c}\cap A_a=\emptyset$
(\emph{empty not by construction}).

\begin{figure}\centering
\includegraphics{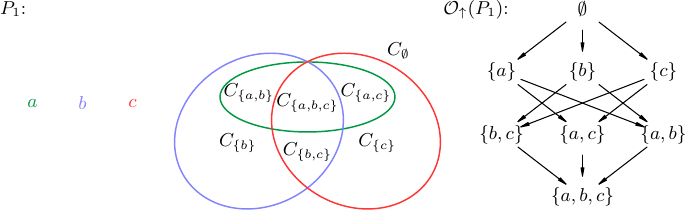}
\caption{
Example for a class ($C_{\{a\}}$), which is empty, but not by construction.
(Compare with the first row of figure~\ref{fig:classVenn})
}\label{fig:classVennX}
\end{figure}

Earlier version of these results was shown in~\cite{Szalay-2015b} 
in the special setting where it was used ($P=P_\text{II*}$).
Note that the present formulation is more general,
here $P$ does not have to be a lattice,
and the sets $A_x$ do not have to cover $A$ entirely ($\bigcup_{x\in P}A_x\subseteq A$).

\section{Multipartite quantum states}
\label{appsec:qstates}

\subsection{Order isomorphisms for Level~I-II}
\label{appsec:qstates.I-II}

\begin{proofl}{of~\eref{eq:DincI} and~\eref{eq:DincII}}
The \textit{first inclusion} in~\eref{eq:DincI},
$\upsilon\preceq\xi \inlineiff 
\mathcal{D}_{\upsilon\text{-unc}}\subseteq\mathcal{D}_{\xi\text{-unc}}$,
was proven in Appendix A.4 in~\cite{Szalay-2015b} 
for \emph{pure} $\xi$-separable (hence pure $\xi$-uncorrelated) states only.
For mixed $\xi$-uncorrelated states, a slight modification is needed.\\
To see the \textit{``only if''} implication,
let us have $\varrho\in\mathcal{D}_{\upsilon\text{-unc}}$, then 
$\varrho=\bigotimes_{Y\in\upsilon}\varrho_Y
=\bigotimes_{X\in\xi}\Bigl(\bigotimes_{Y\in\upsilon, Y\subseteq X}\varrho_Y\Bigr) \in\mathcal{D}_{\xi\text{-unc}}$,
where the \textit{first equality} is 
\eref{eq:DuncI};
and at the \textit{second equality}
we have used the assumption $\upsilon\preceq\xi$, 
which gives by definition that $\forall Y\in\upsilon$,  $\exists X\in\xi$ such that $Y\subseteq X$,
making possible to collect the states of subsystems $Y$ contained in a given $X$,
which can be done for all subsystems $X$.\\
To see the \textit{``if''} implication,
we prove the contrapositive statement,
$\upsilon\not\preceq\xi \inlinethen \mathcal{D}_{\upsilon\text{-unc}}\not\subseteq\mathcal{D}_{\xi\text{-unc}}$.
Let us have $\varrho\in\mathcal{D}_{\upsilon\text{-unc}}$, then, 
using the notation $\varrho_X=\Tr_{L,X}\varrho$, consider
$\bigotimes_{X\in\xi}\varrho_X
= \bigotimes_{X\in\xi}\Tr_{L,X}\varrho 
= \bigotimes_{X\in\xi}\Tr_{L,X}\bigotimes_{Y\in\upsilon}\varrho_Y
= \bigotimes_{X\in\xi}\bigotimes_{Y\in\upsilon}\Tr_{Y,X\cap Y}\varrho_Y
= \bigotimes_{Y\in\upsilon} \Bigl( \bigotimes_{X\in\xi}\Tr_{Y,X\cap Y}\varrho_Y \Bigr)
\neq \bigotimes_{Y\in\upsilon}\varrho_Y=\varrho$,
where at the \textit{second and the last equalities} we used the assumption that $\varrho\in\mathcal{D}_{\upsilon\text{-unc}}$
(we use the notation $\Tr_{X,X'}=\bigotimes_{i\in X\cap \cmpl{X'}}\Tr_{\mathcal{H}_i}:\Lin\mathcal{H}_X\to\Lin\mathcal{H}_{X'}$ for the partial trace,
when $X'\subseteq X$);
the \textit{third equality} can be checked 
by the decomposition of tensors into linear combination of elementary tensors,
and using the linearity of the partial trace and the tensor product;
the \textit{fourth equality} is just the associativity of the tensor product.
The \textit{nonequality} comes from the assumption that $\upsilon\not\preceq\xi$,
which gives that $\exists Y\in\upsilon$ for which $\forall X\in\xi$ we have $Y\not\subseteq X$,
then 
the term $\bigotimes_{X\in\xi}\Tr_{Y,X\cap Y}\varrho_Y\neq \varrho_Y$ for this $Y$,
if $\varrho_Y$ is not of the product form, which is an extra assumption, which can be fulfilled,
since $\dim\mathcal{H}_i>1$.\\
The \textit{second inclusion} in~\eref{eq:DincI},
$\upsilon\preceq\xi \inlineiff 
\mathcal{D}_{\upsilon\text{-sep}}\subseteq\mathcal{D}_{\xi\text{-sep}}$,
has already been proven in Appendix A.4 in~\cite{Szalay-2015b}.\\
The \textit{first inclusion} in~\eref{eq:DincII},
$\vs{\upsilon}\preceq\vs{\xi} \inlineiff 
\mathcal{D}_{\vs{\upsilon}\text{-unc}}\subseteq\mathcal{D}_{\vs{\xi}\text{-unc}}$,
was proven in Appendix A.8 in~\cite{Szalay-2015b} 
for \emph{pure} $\vs{\xi}$-separable (hence pure $\vs{\xi}$-uncorrelated) states only.
For mixed $\vs{\xi}$-uncorrelated states, the same steps can be applied.\\
The \textit{second inclusion} in~\eref{eq:DincII},
$\vs{\upsilon}\preceq\vs{\xi} \inlineiff 
\mathcal{D}_{\vs{\upsilon}\text{-sep}}\subseteq\mathcal{D}_{\vs{\xi}\text{-sep}},$
has already been proven in Appendix A.8 in~\cite{Szalay-2015b}.
\end{proofl}

Note that~\eref{eq:DincI} and~\eref{eq:DincII} immediately lead to that
$\upsilon=\xi$ if and only if $\mathcal{D}_{\upsilon\text{-unc}}=\mathcal{D}_{\xi\text{-unc}}, 
\mathcal{D}_{\upsilon\text{-sep}}=\mathcal{D}_{\xi\text{-sep}}$,
while
$\vs{\upsilon}=\vs{\xi}$ if and only if $\mathcal{D}_{\vs{\upsilon}\text{-unc}}=\mathcal{D}_{\vs{\xi}\text{-unc}}, 
\mathcal{D}_{\vs{\upsilon}\text{-sep}}=\mathcal{D}_{\vs{\xi}\text{-sep}}$,
since an order isomorphism is automatically bijective.

\end{document}